\documentclass[12pt]{article}
\usepackage{amsmath}
\usepackage{amssymb}
\usepackage{amsthm}
\usepackage{fullpage}
\usepackage[utf8]{inputenc}
\usepackage{mathtools}
\usepackage{url}
\usepackage[usenames,svgnames]{xcolor}
\usepackage{cite}
\usepackage{enumitem}
\usepackage{datetime}
\usepackage[titletoc,title]{appendix}
\usepackage{comment}
\usepackage{mathdots}
\usepackage{graphicx}

\usepackage[pagebackref=false]{hyperref} 
\usepackage[all]{hypcap} 

\theoremstyle{plain}
\newtheorem{theorem}{Theorem}[section]
\newtheorem{lemma}[theorem]{Lemma}

\newtheorem{proposition}[theorem]{Proposition}

\newtheorem{example}{Example}[section]
\newtheorem{examples}{Example}[subsection]

\newtheorem{remark}{Remark}[section]

\theoremstyle{definition}
\newtheorem{definition}{Definition}[section]

\numberwithin{equation}{section} 

\hypersetup{
breaklinks,
colorlinks,
citecolor=Green,
linkcolor=BlueViolet,
urlcolor=DarkBlue,
pdftitle={Weighted Hurwitz numbers , constellations and topological recursion},
pdfauthor={A. Alexandrov, G. Chapuis, B. Eynard, J. Harnad } }

\newcommand{\mat}[1]{\mathbf{#1}}

\DeclareMathOperator{\cyc}{cyc}

\DeclareMathOperator{\aut}{aut}
\DeclareMathOperator{\Res}{Res}

\def\ra{{\rightarrow}}

\def\Tr{\mathrm {Tr}}

\def\det{\mathrm {det}}

\def\span{\mathrm {span}}

\def\ln{\mathrm {ln}}

\def\span{\mathrm {span}}

\DeclarePairedDelimiter{\no}{:}{:}

\def\be{\begin{equation}}
\def\ee{\end{equation}}

\def\bea{\begin{eqnarray}}
\def\eea{\end{eqnarray}}

\def\bt{\begin{theorem}}
\def\et{\end{theorem}}

\def\bex{\begin{example}\small \rm}
\def\eex{\end{example}}

\def\bexs{\begin{examples}\small \rm}
\def\eexs{\end{examples}}

\def\ra{\rightarrow}
\def\ss{\subset}
\def\deq{\coloneqq}

\def\br{\begin{remark}\small \rm}
\def\er{\end{remark}}

\def\&{&{\hskip -20pt}}

\def\OO{\mathcal{O}}
\def\PP{\mathcal{P}}
\def\QQ{\mathcal{Q}}

\def\WW{\mathcal{W}}

\def\Cb{\mathbf{C}}

\def\Ib{\mathbf{I}}

\def\Nb{\mathbf{N}}

\def\Nb{\mathbf{N}}
\def\Pb{\mathbf{P}}

\def\Rb{\mathbf{R}}

\def\Vb{\mathbf{V}}

\def\Zb{\mathbf{Z}}

\begin{document}
\baselineskip 16pt
\medskip
\begin{center}
\begin{Large}\fontfamily{cmss}
\fontsize{17pt}{27pt}
\selectfont
\textbf{Weighted Hurwitz numbers and topological recursion: an overview}
\end{Large}
\\
\bigskip \bigskip
\begin{large}  A. Alexandrov$^{1, 2, 3}$\footnote{e-mail: alexandrovsash@gmail.com},   G. Chapuy$^{4}$\footnote{e-mail: guillaume.chapuy@liafa.univ-paris-diderot.fr}, 
B. Eynard$^{2, 5}$\footnote{e-mail:  bertrand.eynard@cea.fr}  and J. Harnad$^{2, 6}$\footnote {e-mail: harnad@crm.umontreal.ca  }
 \end{large}\\
\bigskip
\begin{small}
$^{1}${\em Center for Geometry and Physics, Institute for Basic Science (IBS), Pohang 37673, Korea}\\
 \smallskip
$^{2}${\em Centre de recherches math\'ematiques,
Universit\'e de Montr\'eal\\ C.~P.~6128, succ. centre ville, Montr\'eal,
QC H3C 3J7,  Canada}\\
 \smallskip
 $^{3}${\em ITEP, Bolshaya Cheremushkinskaya 25, 117218 Moscow, Russia} \\
 \smallskip
$^{4}${\em CNRS UMR 7089, Universit\'e Paris Diderot \\
 Paris 7, Case 7014, 75205 Paris Cedex 13, France}\\
 \smallskip
 $^{5}${\em Institut de Physique Th\'eorique, CEA, IPhT \\
 CNRS   URA 2306, F-91191 Gif-sur-Yvette, France} \\
 \smallskip
$^{6}${\em Department of Mathematics and Statistics, Concordia University\\ 1455 de Maisonneuve Blvd.~W.~Montreal, QC H3G 1M8,  Canada}\\
\smallskip

\end{small}
\end{center}

\begin{abstract}
 Multiparametric families of hypergeometric  $\tau$-functions  of KP or Toda type  serve as generating functions 
 for weighted Hurwitz numbers, providing weighted enumerations of branched covers of the Riemann sphere. 
  A graphical  interpretation of the weighting  is given in terms of constellations mapped onto the covering surface.
  The theory is placed within the framework of topological recursion, with the Baker function at ${\bf t} ={\bf 0}$ 
 shown to satisfy the quantum spectral curve equation, whose classical limit is rational.  A basis 
 for the space of formal power series  in the spectral variable is generated that is adapted to the Grassmannian 
 element associated to the $\tau$-function. Multicurrent correlators are defined in terms of the $\tau$-function and 
 shown to  provide an alternative generating function for weighted Hurwitz numbers.  Fermionic VEV representations 
 are provided for the adapted  bases, pair correlators and multicurrent correlators.   
 Choosing the weight generating function as a polynomial,  and  restricting the number of  nonzero ``second'' KP flow 
 parameters  in the Toda $\tau$-function to be finite implies a finite rank covariant derivative equation with rational coefficients
satisfied by a finite ``window'' of adapted basis elements.  The pair correlator is shown to
provide a Christoffel-Darboux type finite rank integrable kernel, and the WKB series coefficients
of the associated adjoint system are computed recursively, leading to topological recursion relations
for the generators of the weighted Hurwitz numbers.

     \end{abstract}

\tableofcontents

\break

\section{Weighted Hurwitz numbers and  generating functions }

The study of weighted enumerations of branched coverings of the Riemann sphere
was initiated by Hurwitz \cite{Hur}. It was subsequently related to the character theory of the symmetric group
by Frobenius \cite{Frob}, and  computations of these enumerative invariants were
obtained for low genus and simple branching structures. Interest was more recently revived by the work of Okounkov and Pandharipande \cite{Ok,Pa}, who showed that certain special forms of  the Sato-Segal-Wilson \cite{SS, Sa, SW} KP $\tau$-function 
and 2D-Toda $\tau$-functions, central to the modern theory of integrable systems,  serve as generating functions for ``simple'' Hurwitz numbers, which enumerate branched covers in which all but one or two of the specified branching profiles are $2$-cycles. 

Further special cases, relating to the enumeration of  monotonic paths in the Cayley graph of the symmetric group \cite{GGN1, GH1},
three point branched covers \cite{AC1, AC2, KZ, Z,AMMN2}, coverings of $\Rb\Pb^2$  \cite{NaOr} and certain cases of matrix integrals \cite{BEMS,AC1,MSh,AM} were subsequently considered, and their generating functions similarly shown  to be special cases of $\tau$-functions of hypergeometric type \cite{OrSc1, OrSc2, Or}. The most general case of weighted Hurwitz numbers, with weight generating functions depending on an infinite number of weighting parameters, was developed in  \cite{GH1} - \cite{HO}. 

The enumeration of branched coverings is classically equivalent to the enumeration of {\it constellations}, which are special types of bipartite graphs mapped onto Riemann surfaces and ``decorated'' in several possible ways \cite{LZ}. The enumerative study of maps with respect to their genus invariant has been considerably developed by combinatorists since the pioneering work of Tutte in the planar case~\cite{Tutte} and Bender and Canfield~\cite{BC} for higher genus. (The reader is referred to~\cite{Schaeffer} for an accessible introduction, historical references, and a modern entry to the rapidly growing literature on this topic.)

In recent years, the topological recursion program \cite{EO1, EO2, EO3}  has been a significant development  in the study of enumerative, combinatorial, geometric and topological invariants associated to maps, Riemann surfaces, moduli spaces
and knots. It has led to an explicit  finite recursive algebraic scheme for computing such invariants.
This begins with the construction of an associated ``quantum curve'' which, in general, is an ordinary differential or difference
equation in the spectral parameter satisfied by the associated Baker function, whose WKB formal series solution,
expanded in a suitable ``small parameter'' is studied recursively. Partial results on the inclusion of Hurwitz numbers in this
scheme have been obtained in  \cite{ALS, BEMS, KZ, MSS, MS, DLN, DDM}, mainly for the special cases mentioned above. In the 
present work, it is shown how a broad  class of multiparametrically weighted Hurwitz numbers can be systematically placed  into the topological recursion program,  and the associated spectral curve, both quantum and classical, determined explicitly in terms of the
weight generating function. 

The aim of this work is to provide an accessible introduction that summarizes the main results in this direction.
All relevant constructions are indicated, and the principal results stated. Some  proofs are provided in detail; others are only
sketched in outline, with full details provided in one of two companion publications \cite{ACEH1, ACEH2}. 
 The results include: explicit parametrization of both the classical and quantum spectral curves;  the system of  recursion relations satisfied  by dual pairs of adapted bases for the underlying space of formal power series  in a spectral parameter; construction of  suitably defined multipair correlators and multicurrent correlators, which serve as generating functions for weighted Hurwitz numbers; expressions for all relevant quantities as fermionic vacuum state expectation values; a finite rank expression for the Fredholm integral kernel associated to the pair correlator, analogous to the Christoffel-Darboux  kernel appearing in the theory of orthogonal polynomials and, finally, the topological recursion relations  satisfied by the multiform weighted Hurwitz number generators.

\subsection{Geometric and combinatorial definitions of  Hurwitz numbers}
\label{pure_hurwitz_def}

Given a set of $k$ partitions  of $N\in \Nb^+$, $\{\mu^{(i)}\}_{i=1, \dots, k}$  , $|\mu^{(i)}|=N$, we define the {\it pure Hurwitz number}
$H(\mu^{(1)}, \dots , \mu^{(k)})$ geometrically as  the number of inequivalent $N$-sheeted branched coverings of the Riemann sphere with $k$ branch points having ramification profiles $\{\mu^{(i)}\}_{i=1, \dots, k}$, normalized by the inverse of the order of the automorphism group of the cover. 
The Riemann-Hurwitz theorem expresses the Euler characteristic of the covering surface as
\be
\chi = 2-2g = 2N - \sum_{i=1}^k \ell^*(\mu^{(i)}),
\ee
where $\ell^*(\mu)$ denotes the {\it colength}  of the partition $\mu$ (i.e.  the complement 
of its length)
\be
\ell^*(\mu):= |\mu| - \ell(\mu).
\ee

Combinatorially, it may equivalently be defined as the number of distinct ways the identity element $\Ib \in S_N$ in the symmetric
group $S_N$ may be factorized
\be
\Ib = h_1 \cdots h_k
\label{Id_factoriz}
\ee
as a product of $k$ elements $h_i \in \cyc(\mu^{(i)})$ in the conjugacy classes
of cycle type $\mu^{(i)}$; i.e., whose cycle lengths are the parts of the partition $\mu^{(i)}$, multiplied by the normalization factor $1/N!$.
The Frobenius-Schur formula \cite{LZ}[Appendix A] expresses this in terms of the irreducible group characters  $\chi_\lambda(\mu)$ of $S_N$. 
\be
H(\mu^{(1)}, \dots , \mu^{(k)}) = \sum_{|\lambda|=n}h_\lambda^{k-2} \prod_{i=1}^k {\chi_\lambda(\mu^{(i)}) \over z_{\mu^{(i)}}}.
\ee
Here $h_\lambda$ is the product of the hook lengths of  the Young diagram of the partition $\lambda$  determining the irreducible representation and $z_\mu$ is the order of the stabilizer of an element of the conjugacy class of cycle type $\cyc(\mu)$. 

The equivalence of the two definitions follows from the homomorphism from the fundamental
group of the sphere punctured at the branch points to the symmetric group $S_N$ generated by
monodromy, which equates the identity element to the product of monodromies of paths consisting of  
simple loops around a suitably ordered sequence of all branch points.

\subsection{Weighted Hurwitz numbers }
\subsubsection{Weight generating function $G(z)$ }
\label{weight_gen_G}
Following refs.~\cite{GH1} -\cite{HO}, we define multiparametric weighted Hurwitz numbers by first  introducing
a generating function $G(z)$ for the weights. This  admits either a formal power
series representation
\be
G(z) = 1 + \sum_{i=1}^\infty g_i z^i
\label{G_inf_sum}
\ee
or an infinite product one
\be
G(z) = \prod_{i=1}^\infty (1+c_i z), 
\label{G_inf_prod}
\ee
(or a limiting form of the latter),  where ${\bf c}= (c_1, c_2, \dots )$ denotes an infinite sequence of 
indeterminates that may be evaluated to real or complex constants. If the two  expansions are
compared, with (\ref{G_inf_sum}) viewed as a formal series, this may be interpreted as the generating function
 for the elementary symmetric functions in the indeterminates $(c_1, c_2, \dots )$.
\be
g_i = e_i ({\bf c}).
\ee

\subsubsection{Geometrical definition: weighted coverings}
\label{geometric}

\begin{definition}
Picking a pair of partitions $(\mu, \nu)$ of $N$, and a positive integer $d\in \Nb^+$, the {\it geometrically weighted
 single and double Hurwitz numbers}  $H_{G}^d(\mu)$,  $H_{G,}^d(\mu, \nu)$ associated with the generating function 
 $G$ are defined  to be the sums
   \bea
   H^d_G(\mu) &\&\deq \sum_{k=0}^d \sideset{}{'}\sum_{\substack{\mu^{(1)}, \dots, \mu^{(k)} \\ |\mu^{(i)}| = N \\ \sum_{i=1}^k \ell^*(\mu^{(i)})= d}}
\WW_G(\mu^{(1)}, \dots, \mu^{(k)})H(\mu^{(1)}, \dots, \mu^{(k)}, \mu) ,
\label{Hd_G_mu} \\
H^d_G(\mu, \nu) &\&\deq \sum_{k=0}^d \sideset{}{'}\sum_{\substack{\mu^{(1)}, \dots, \mu^{(k)} \\ |\mu^{(i)}| = N \\ \sum_{i=1}^k \ell^*(\mu^{(i)})= d}}
\WW_G(\mu^{(1)}, \dots, \mu^{(k)})H(\mu^{(1)}, \dots, \mu^{(k)}, \mu, \nu) ,
\label{Hd_G_mu_nu}
\eea
where $\sideset{}{'}\sum$ denotes a sum over all  $k$-tuples of partitions 
$\{\mu^{(1)}, \dots, \mu^{(k)}\}$ of $N$ other than the cycle type of the identity element $(1^N)$ and
 the weights $\WW_G(\mu^{(1)}, \dots, \mu^{(k)})$ are given by
 \be
 \WW_G(\mu^{(1)}, \dots, \mu^{(k)}):=
 {1\over k!}\sum_{\sigma \in \mathfrak{S}_{k}} \sum_{1 \le b_1 < \cdots < b_{k }} 
 c_{b_{\sigma(1)}}^{\ell^*(\mu^{(1)})} \cdots c_{b_{\sigma(k)}}^{\ell^*(\mu^{(k)})} = {|\aut(\lambda)|\over k!} m_\lambda ({\bf c}).
 \label{Wg_def}
 \ee
Here  
\be
m_\lambda ({\bf c}) = {1\over |\aut(\lambda)|}\sum_{\sigma \in \mathfrak{S}_{k}} \sum_{1 \le b_1 < \cdots < b_{k }}
 c_{b_{\sigma(1)}}^{\lambda_1} \cdots c_{b_{\sigma(k)}}^{\lambda_{k}} ,
  \label{m_lambda}
\ee
is the monomial sum symmetric function \cite{Mac} of the indeterminates ${\bf c}$
corresponding to the partition $\lambda$  of length $k$ and weight $|\lambda|= d$
whose parts $\{\lambda_i\}$ are the colengths $\{\ell^*(\mu^{(i)})\}$, written in weakly descending order,
\be
\{\lambda_i\}_{i=1, \dots k} \sim \{\ell^*(\mu^{(i)})\}_{i=1, \dots k}
\ee
with normalization factor 
\be
|\aut(\lambda)| := \prod_{i} m_i(\lambda)! 
\ee
equal to the order of the automorphism group of the partition $\lambda$, 
where $m_i(\lambda)$ is the number of parts of $\lambda$ equal to~$i$. 
\end{definition}

 If the sums in (\ref{Hd_G_mu}) and (\ref{Hd_G_mu_nu})
are taken over connected coverings only  (or, equivalently, the factors $(h_1, \dots, h_k)$ are required to generate
a subgroup of $S_N$ that acts transitively), we denote the corresponding sums 
 $\tilde{H}^d_G(\mu)$ and $\tilde{H}^d_G(\mu, \nu)$, respectively.

\subsubsection{Combinatorial definition: weighted constellations}
\label{weighted_hurwitz_def}

Another  way to arrive at weighted Hurwitz numbers is through weighted 
{\it constellations} \cite{LZ, BMS, Ch}. In \cite{LZ} a constellation of length $k$ and degree $N$
is defined  as a factorization (\ref{Id_factoriz}) of the identity element in $\Ib \in S_N$ into a product $h_1\cdots h_k$
such that the group $\langle h_1, \dots, h_k\rangle$ they generate acts transitively
on the set $(1, \dots, N)$. Here, we consider an enhanced version,  in which a weighting
as assigned to all vertices and edges, and the total weight of the constellation is the product of these.
We further consider $k+2$ constellations of  degree (or weight) $N$,  since these correspond
to  $N$-sheeted branched covers of the Riemann sphere in  ramification points over
$k$ of the branch points are assigned a distinct weight in the enumeration, while two of them 
(say, over $0$ and $\infty$) are weighted differently in the interpretation of the $\tau$-function (\ref{tau_G_beta_gamma}),
(\ref{tau_G_tilde_H}) below as a generating function.
   The graph theoretic equivalent of a $k+2$-constellation of degree (or weight)  $N$ is a bipartite graph on a Riemann surface 
whose two types of vertices are either  {\it coloured} or  {\it star} vertices, with all faces homeomorphic to a disc.
(Thus, it is a  {\it map} with certain additional special structure.) There are $N$ star  vertices, 
numbered consecutively from $1$ to $N$ and at most $(N-1)(k+2)$ coloured vertices carrying
one of $k$ colours, labelled by distinct positive integers $b_i$ with $i$  ranging from from $1$ to $k$, 
plus a further two, labelled $0$ and $\infty \equiv k+1$, denoting, in our adaptation, the special case
of ``white'' and ``black'' vertices respectively. Each star vertex has an edge connecting it to all $k+2$ coloured vertices, consecutively
labelled counter-clockwise. The number of coloured vertices of given colour $i$ is equal to the number 
of parts $\ell(\mu^{(i)})$ of the partition $\mu^{(i)} $ (i.e., its {\it length}), or $\ell(\mu)$ for the ``white'' vertices
and $\ell(\nu)$ for the ``black'' ones.  As an illustration, see the example of Figure~\ref{fig:exampleConstellationBis} 
where $N=5, k=3$.

 \begin{figure}
\begin{center}
\includegraphics{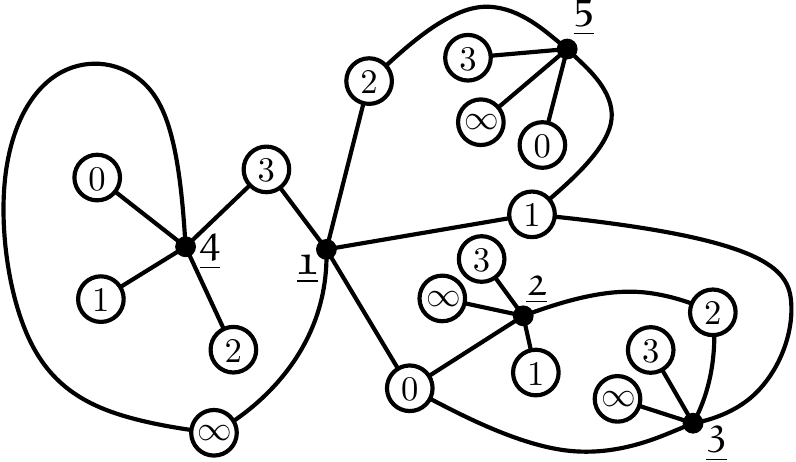}
\end{center}
\caption{\footnotesize  An example of a  constellation with $N=5$, $k=3$\label{fig:exampleConstellationBis}, corresponding to the values $ N =5,  \quad k=3$ and to the factorization $h_0h_1h_2h_3h_4=1$ with 
$h_1  = (135), \quad h_2 =  (15) (23), \quad h_3 = (14),  
   h_0 = (321),  \quad   h_4 = h_\infty = (14)$, with corresponding partitions 
 $\mu^{(1)} = (3, 1, 1)),  \quad \mu^{(2)} = (2, 2, 1), \quad  \mu^{(3)} = (2, 1,1,1),\ 
\mu:= \mu^{(0)}  = (3, 1, 1) ), \quad   \nu := \mu^{(4)}= (2, 1, 1, 1)
	 $. }
\end{figure}

Associated to each coloured vertex  is a cycle in $S_N$ obtained by listing the star vertices
connected to it by an edge in counterclockwise order. Taking the product of the distinct cyclic permutations 
associated to each colour $i = 0, \dots, k$, we obtain an element $h_i\in S_N$,  with cycle type $\cyc(\mu^{(i)})$, where $\mu^{(i)}$ 
is the partition of $N$ whose parts are the cycle lengths of $h_i$.  The element associated to the white vertices is denoted $h_0=h$
and its cycle type denoted $\mu:= \mu^{(0)}$, the one associated to the black vertices whose cycle type is denoted $\nu:= \mu^{(k+1)}$
 is $h_{k+1}$ and the other coloured vertices are associated to the elements $\{h_i\}_{i=1, \dots, k}$,  with cycle types $\{\mu^{(i)}\}_{i=1, \dots, k}$
respectively, and  satisfy
\be
\Ib = h_0 h_1 \cdots h_{k+1}
\label{h_factorization}
\ee

To obtain the constellation associated with a particular branched cover of $\Cb\Pb^1$, we place a coloured vertex
of colour $i$ at each of the ramification points $\{Q^{(i)}_j\}_{j=1 , \dots \ell(\mu^{(i)})}$ over a given branch point 
$\{Q^{(i)}\}_{i=0, \dots , k+1}$, and choose a single generic (non-branched) 
point $P\in \Cb\Pb^1$ over which we have, in some specified order, the $N$ points $(P_1, \dots, P_N)$, which are identified as the
star vertices. For every simple loop starting and ending at $P_i$, going round one of the branch points, we have the
monodromy lift on the covering surface which permutes the points $\{P_i\}$   in such a way that each ramification point of colour $i$ defines
a disjoint cycle, with the product of all cycles over a given branch point $Q^{(i)}$ equal to the monodromy element $h_i$. 
The edges connect the ramification points to those $P_i$'s that belong to the given cycle.
(See Figure~\ref{fig:pathABCBis}.)


 \begin{figure}
\begin{center}
\includegraphics{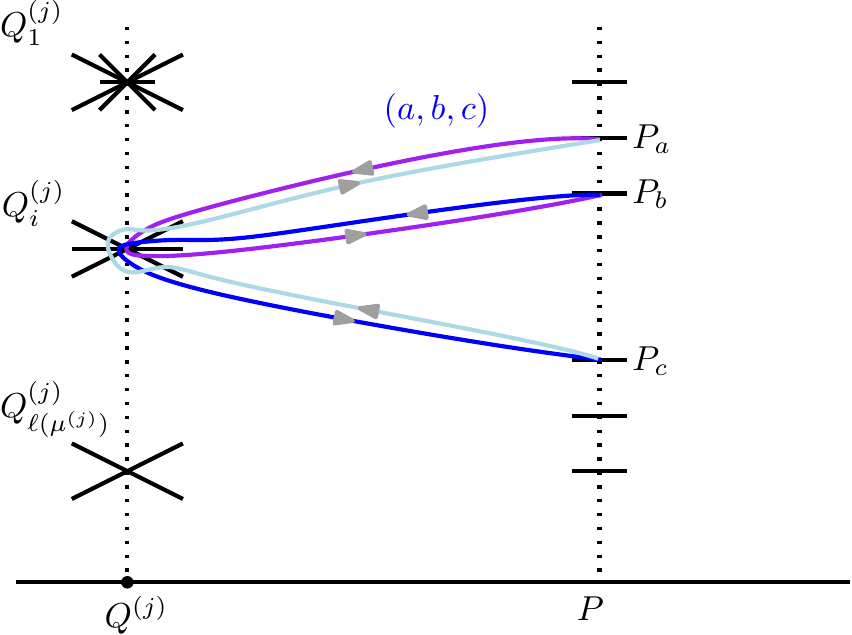}
\end{center}

\caption{\footnotesize Lifted loops and edges }\label{fig:pathABCBis}

\end{figure}

In order to define the weighting, we add a further pair of indeterminates $(\beta, \gamma)$, 
 that serve as expansion parameters, and may, for weighting purposes, be evaluated to real or complex numbers.
 We then assign, to each (non-white, non-black) coloured  vertex $Q^{(i)}_j$ of colour  $b_i$,  a weight $\beta^{-1} c_{b_i}^{-1}$; 
 to each edge connecting  non-white, non-black vertex  $Q^{(i)}_j$ to a star vertex, the weight  $\beta c_{b_i}$
 to each white vertex, a weight 
\be
p_{\mu_j} = \mu_j t_{\mu_j} , \quad j=1, \dots, \ell(\mu)
\label{com1}
\ee
equal to the power sum symmetric function in an arbitrary number of auxiliary variables, where 
$\mu_i$ is equal to the number of edges connecting that white vertex to star vertices;
 to each black vertex, a weight 
\be
p_{\nu_j} = \nu_j s_{\nu_j} , \quad j=1, \dots, \ell(\nu)
\label{com1}
\ee
equal to the power sum symmetric function in a second set of auxiliary variables, where 
$\nu_i$ is equal to the number of edges connecting that black vertex to star vertices;
to the edges connecting the white or black ramification points $\{ Q^{(0)}_j, Q^{k+1)}_j\}$
to any  star vertex, the weight $1$, and to each of the star vertices, the weight 
\be
\gamma
\ee
The doubly infinite set of parameters ${\bf t} =(t_1, t_2, \dots)$ 
and ${\bf s} = (s_1, s_2, \dots)$  are interpreted as the two independent KP flow parameters that appear 
in the $2D$-Toda $\tau$-function \cite{SS, SW, Takeb, UTa}.

Taking the product of these  over all vertices, edges and faces gives the overall weight
\be
{\gamma^N \beta^d \over N!}  \prod_{i=1}^k c_{b_i}^{\ell^*(\mu^{(i)})} p_\mu ({\bf t}) p_\nu ({\bf s})
\ee
for the constellation, where the standard normalization factor ${1\over N!}$ is included in the
enumeration, 
\be
d := \sum_{i=1}^k \ell^*(\mu^{(i)})
\ee
is related to the Euler characteristic of the surface by the Riemann-Hurwitz relation
\be
2 - 2g = 2N - d
\ee
and
\be
p_\mu({\bf t}) := \prod_{i=1}^{\ell(\mu)} p_{\mu_i}, \quad p_\nu({\bf s}) := \prod_{i=1}^{\ell(\nu)} p'_{\nu_i}.
\ee
Summing these over all such weighted constellations, with all possible choices of $k$ distinct colours, 
chosen from amongst the indices labelling the nonzero constants $c_i$ in (\ref{G_inf_prod}), 
corresponding to a factorization (\ref{h_factorization}) with the same cycle types, permuted in all ways, 
 normalized by ${1\over |\aut (\lambda)|}$, where $\lambda$ is the partition of weight
\be
|\lambda| = d:= \sum_{i=1}^k \ell^*(\mu^{(i)})
\label{lambda_def}
\ee
whose parts are the colengths $\{\ell^*(\mu^{(i)})\}$, suitably ordered, and summing over the integer $k$,
we obtain the total weighting factor $m_\lambda (c_1, c_2, \dots )$ of (\ref{m_lambda}), multiplied
by the Hurwitz number $H(\mu, \mu^{(1)}, \dots, \mu^{(k)}, \nu)$ which enumerates the number
of products (\ref{h_factorization}) with factors in the conjugacy classes $(\mu, \mu^{(1)}, \dots, \mu^{(k)}, \nu)$, and recover
the $\tau$-function (\ref{tau_G_H}) as generating function, as explained below.

\section{Hypergeometric $\tau$-function, Baker function,  recursions and correlators}
\label{tau_baker_recursion}
\subsection{Hypergeometric $\tau$-functions as generating functions for weighted Hurwitz numbers}
\label{tau_functions}

Following refs.~{\cite{GH1, GH2, H1, HO}, we define  a family of 2D Toda $\tau$-functions
of hypergeometric type \cite{KMMM, OrSc1, OrSc2} associated to the generating function $G$ and  a pair $(\beta, \gamma)$ of complex parameters by the  double Schur function \cite{Mac} expansion:
\be
  \tau^{(G, \beta, \gamma)} ({\bf t}, {\bf s})  = \sum_{\lambda} \gamma^{|\lambda|} r^{(G,\beta)}_\lambda s_\lambda ({\bf t}) s_\lambda ({\bf s}), 
  \label{tau_G_beta_gamma}
\ee
where the Schur functions $s_\lambda({\bf t}), s_\lambda({\bf s})$ are
viewed as functions of two infinite sequences of 2-Toda flow variables ${\bf t}=(t_1, t_2, \dots)$, ${\bf s}=(s_1, s_2, \dots)$ 
and the {\em content product} coefficients $r_\lambda^{(G, \beta)} $ are defined in terms of the weight 
generating function $G$ as
\be
r_\lambda^{(G,\beta)}\deq   \prod_{(i,j)\in \lambda} r_{j-i}^{G}(\beta),
\label{r_lambda_G}
\ee
where
\be
r_j^{G}( \beta):= G(j \beta ).
\ee
Denote by $\mathbb K = \mathbb Q[g_1,\dots, g_M, \cdots] $
the algebra of polynomials in the $g_k$'s, with rational coefficients or, equivalently,
the algebra of symmetric functions of the $c_i$'s.
In the following, it will also be useful to define  the quantities $\{\rho_j, T_j\}_{j\in \Zb}$ by
\bea
 r_j^{G}(\beta)&\& = {\rho_j \over \gamma \rho_{j-1}},\\
\rho_j &\&:= \gamma^j \prod_{i=1}^j G(i\beta) =  e^{T_j} , \quad \rho_0 =1, 
\label{rho_j+}
\\
 \rho_{-j} &\&:= \gamma^{-j}\prod_{i=0}^{j-1} (G(-i \beta))^{-1} = e^{T_{-j}} , \quad j=  1, 2 \dots. \label{rho_j_gamma_G} 
 \label{rho_j-}
\eea
Moreover, it is easily seen that there exists a formal power series $T(x)$ such that
\be
e^{T(x)-T(x-1)}=\gamma G(\beta x),
\ee
with $T(0)=0$, $T(i)=T_i$. 

As shown in refs.~\cite{GH1, GH2, H1, HO}, identifying the flow parameters  as
 \bea
t_i &\&= {p_i\over i}, \quad s_i = {p'_i\over i},\cr
  p_\mu ({\bf t}) &\&:= \prod_{i=1}^{\ell(\mu)}p_{\mu_i}  , \quad p_\nu({\bf s}) := \prod_{i=1}^{\ell(\nu)} p'_{\nu_i} 
\eea
and changing to the basis of power sum symmetric functions \cite{Mac},
the $\tau$-function (\ref{tau_G_beta_gamma}) serves as a generating function for the weighted Hurwitz numbers 
defined above.
\begin{theorem}
\be
\tau^{(G, \beta, \gamma)} ({\bf t}, {\bf s})  = \sum_{d=0}^\infty \beta^d \sum_{\substack{\mu, \nu \\ |\mu|=|\nu|}} \gamma^{|\mu|}  H^d_G(\mu, \nu) p_\mu({\bf t}) p_\nu({\bf s}).
\label{tau_G_H}
\ee
\end{theorem}
By standard combinatorial arguments, taking its logarithm gives the generating function for connected weighted Hurwitz numbers:
\be
\ln(\tau^{(G, \beta, \gamma)} ({\bf t}, {\bf s}) ) = \sum_{d=0}^\infty \beta^d \sum_{\substack{\mu, \nu \\ |\mu|=|\nu|}} \gamma^{|\mu|}  \tilde{H}^d_G(\mu, \nu) p_\mu({\bf t}) p_\nu({\bf s}).
\label{tau_G_tilde_H}
\ee
All instances of weighted Hurwitz numbers to have appeared in the literature   \cite{AEG, AC1, AC2, AMMN1,AMMN2, BEMS, GH1, GH2, GGN1, H1, H2, HO, KZ, LZ,  Ok, Pa, Z} have $\tau$-functions of hypergeometric type  as generating functions and are either special cases of the above or are similarly derived from a second class of ``dual'' weight generating functions
introduced in \cite{GH2,H1}, or a more general quantum weighting introduced  in  \cite{H2}.

\subsection{Baker function, adapted bases and recursions}
 \label{adapted_basis}
 
  Viewed purely as a KP $\tau$-function with respect to the first set of flow variables  ${\bf t} = (t_1, t_2, \dots )$, with the remaining
variables ${\bf s} = (s_1, s_2, \dots )$ interpreted as auxiliary parameters, the corresponding parametric family of Baker functions 
 $\Psi_{(G, \beta, \gamma)}^-(\zeta, {\bf t}, {\bf s})$ and their duals $\Psi_{(G, \beta, \gamma)}^+(\zeta, {\bf t}, {\bf s})$
are given by the Sato formula
 \be
 \Psi_{(G, \beta, \gamma)}^\mp(\zeta, {\bf t}, {\bf s}) = e^{\pm \sum_{i=1}^\infty t_i \zeta^i }{\tau^{(G, \beta, \gamma)}({\bf t} \mp [x] , {\bf s}) \over  \tau^{(G, \beta, \gamma)}({\bf t}, {\bf s}) }
 = e^{\pm \sum_{i=1}^\infty t_i \zeta^i } \left(1 + \OO({1/ \zeta})\right),
 \ee
 where 
 \be
 \zeta := x^{-1}
 \ee
  is the spectral parameter and $[x]$ is the infinite vector whose components are equal to the terms
  of the Taylor expansion of $-\ln(1-x)$ at $x=0$
 \be
 [x]:= (x, {x^2\over 2}, \dots, {x^n \over n}, \dots ).
 \ee
 
 In the following, instead of considering the values of the Baker function and its dual at arbitrary KP flow parameter values, we shall be 
evaluating at ${\bf t}= {\bf 0}$,  and $\beta^{-1}{\bf s}$
\bea
\Psi_{(G, \beta, \gamma)}^+(1/x, {\bf 0}, \beta^{-1}{\bf s}) &\&= \gamma \sum_{j=0}^\infty  \rho_{j-1}h_j(\beta^{-1}{\bf s})  x^j,
\label{Psi_zero}
 \\
\Psi_{(G, \beta, \gamma)}^-(1/x, { \bf  0}, \beta^{-1}{\bf s}) &\&=   \sum_{j=0}^\infty \rho^{-1}_{-j}  h_j(-\beta^{-1}{\bf s}) x^j,
\label{Phi_zero}
\eea
where $\{h_j\}_{j\in \Nb}$ are the complete symmetric functions, defined by
\be
e^{\sum_{i=0}^\infty t_i \zeta^i} = \sum_{j=0}^\infty h_j({\bf t}) \zeta^j.
\ee
More generally, we define the adapted basis $\{\Psi^+_k(x)\}_{k\in \Zb}$ and its dual  $\{\Psi^-_k(x)\}_{k\in \Zb}$ by
\bea
\Psi^+_k(x) :=\gamma  \sum_{j=0}^\infty  \rho_{j+k-1}  h_j(\beta^{-1}{\bf s}) x^{j+k}, 
\label{Psi_plus_k_ser}
\\
\Psi^-_k(x) := \sum_{j=0}^\infty \rho^{-1}_{-j-k}  h_j(-\beta^{-1}{\bf s}) x^{j+k}
\label{Psi_minus_k_ser}
\eea
for all $k\in \Zb $ so that, in particular
\be
\Psi^-_0(x)= \Psi^-_{(G, \beta, \gamma)}(1/x, {\bf 0},  \beta^{-1}{\bf s}), \quad \Psi^+_0(x) = \Psi^+_{(G, \beta, \gamma)}(1/x, {\bf 0}, \beta^{-1}{\bf s}).
\ee
Note that we may view these as formal series in the parameters: 
 \be
 \Psi^+_k(x), \Psi^-_k(x) \in  \gamma^k \mathbb K[x, x^{-1},{\bf s},\beta,\beta^{-1}][[\gamma]].
 \ee
 \br
 The Baker function $\Psi^-_{(G, \beta, \gamma)}(\zeta= 1/x, {\bf t }, \beta^{-1} {\bf s})$
 and its dual $\Psi^+_{(G, \beta, \gamma)}(\zeta= 1/x, {\bf t }, \beta^{-1} {\bf s})$, as defined in 
 (\ref{Psi_minus_k_ser}), (\ref{Psi_plus_k_ser}),  are  elements of the subspaces
 \be
 W^\pm_{(G, \beta, \gamma, {\bf s})} := \span\{ \Psi^\pm_{-k} (\zeta=1/x)\}_{k\in \Nb} \ss
\mathbb K[x, x^{-1}, {\bf s},\beta,\beta^{-1}]((\gamma))
 \ee
 for all values of the KP flow parameters ${\bf t}$. These may be understood as elements of formal power series 
 analogs of the Sato-Segal-Wilson Grassmannian, consisting of subspaces of the Hilbert space spanned
 by powers $\{\zeta^k\}_{k\in \Zb}$ of  the spectral parameter $\zeta = 1/x$, commensurable with the subspace 
 spanned by the monomials  $\{\zeta^k\}_{k=0, 1, \dots}$.
\er

\begin{definition}
Define the recursion operators
\be
R_{\pm}(x):=\gamma x G(\pm\beta D),
\label{Rdef}
\ee
where
\be
D:= x\frac{d}{dx}
\ee
is the Euler operator.
\end{definition}

\begin{lemma} These  satisfy the commutation relations
\bea
\label{Rcom}
\left[D,R_{\pm}\right]=R_{\pm}.
\eea
\end{lemma}
Note also that the operators $R_\pm$ can  be expressed as
\bea
\label{Tcom}
R_+=e^{T\left(D-1\right)}\circ x \circ e^{-T\left(D-1\right)},\\
R_-=e^{-T\left(-D\right)} \circ x \circ e^{T\left(-D\right)}.
\eea
\noindent By application of these operators to the series expansion (\ref{Psi_plus_k_ser}) and (\ref{Psi_minus_k_ser}), 
using (\ref{rho_j+}) and (\ref{rho_j-}),  it follows that the basis elements $\{\Psi^{\pm}_k(x)\}$ satisfy the recursion relations
\begin{proposition}
\bea
\Psi^+_{k\pm 1}(x):=R^{\pm1}_+\Psi^+_k(x),
\label{recursions_Psi_plus_k}
\\
\Psi^-_{k \pm 1} (x):=R^{\pm 1}_- \Psi^-_k(x).
\label{recursions_Psi_minus_k}
\eea
\end{proposition}
From  the identity
\be
\sum_{j=0}^k h_j({\bf t}) h_{k-j}(-{\bf t}) = \delta_{k0},
\ee
for complete symmetric functions, it also follows that the bases consisting of 
functions (\ref{Psi_plus_k_ser}) and (\ref{Psi_minus_k_ser})  are dual to each other 
\be
\left(\Psi^+_k,\Psi^-_m\right) =  \gamma\, \delta_{k+m,1}
\label{dual_Psi_pairing}
\ee
 with respect to the pairing
\be
\left(f,g\right):=\frac{1}{2\pi i} \oint_0 f(x) g(x) \frac{dx}{x^2}.
\label{scalar_product}
\ee

In particular, all $\Psi^+_k$ for $k\leq 0$ are orthogonal to all $\Psi^-_l$ for $l\leq 0$, 
and this is equivalent to the Hirota bilinear relations for  the KP hierarchy with respect to the times ${\bf t}$.

\subsection{The pair correlator and its expansion}
\label{K_tau}

We now define the {\it pair correlation function} $K(x,x')$ 
 in terms of the $\tau$-function.
\begin {definition} The single {\it pair correlator} is defined as 
\begin{eqnarray}
K(x;x') 
&{:=}& \frac{1}{x-x'}\,
\tau^{(G, \beta, \gamma)} \big([x]-[x'],\beta^{-1}{\bf s}\big), 
\label{K_x_x'}
\end{eqnarray}
while more generally,  for $n\geq 1$,  we define the $n$-pair correlator
\be
K_n(x_1,\dots,x_n;x'_1,\dots,x'_n)
:= 
\det\left(\frac{1}{x_i-x'_j}\right)_{1\leq i,j\leq n}
\times
\tau^{(G, \beta, \gamma)} \left(\sum_i [x_i]-[x'_i],\beta^{-1}{\bf s}\right).
\label{K_n_tau_def}
\ee
\end{definition}
As a formal series, for any $n\geq 1$, $K_n (x_1,\dots,x_n;x'_1,\dots,x'_n)$ may be viewed as belonging to $\mathbb K(x_1,\dots,x_n,x'_1,\dots,x'_n)[{\bf s},\beta,\beta^{-1}][[\gamma]]$.
It is a standard result, following from the Cauchy-Binet identity or, equivalently, the fermionic Wick theorem (cf. ref.~\cite{ACEH1}), 
that 
\begin{proposition}
\label{K_n_det}
\be
K_n(x_1,\dots,x_n;x'_1,\dots,x'_n) = \det \left(K(x_i;x'_j)\right).
\ee
\end{proposition}
This should be interpreted as holding in  $\mathbb K(x_1,\dots,x_n,x'_1,\dots,x'_n)[{\bf s},\beta,\beta^{-1}][[\gamma]]$.
We may also express $K(x;x')$ as a formal double Taylor series with coefficients
given by the complete symmetric functions and the weighting parameters $\rho_i$.
\begin{proposition}
\label{K_expansion}
 The function $K(x,x')$, can be expressed as:
\bea
K(x,x') &\&= {1 \over x - x'} +\sum_{a,b \geq 0}  \rho_a \rho^{-1}_{-b-1} x^a (-x')^b s_{(a|b)}(\beta^{-1}{\bf s}) \cr
&\& ={1 \over x - x'} +  \sum_{a=0}^\infty \sum_{b=0}^\infty \sum_{j=1}^{b+1} \
\rho_a h_{a+j} (\beta^{-1}{\bf s}) x^a \rho^{-1}_{-b-1} h_{b-j+1}(-\beta^{-1} {\bf s}) (x')^b.
\label{Khook} 
\eea
\end{proposition} 
\begin{proof} This follows from substituting  the Schur function expansion (\ref{tau_G_beta_gamma}) of the $\tau$-function 
 $\tau^{(G, \beta, \gamma)} ({\bf t}, {\bf s})$, into the definition (\ref{K_x_x'}) of  $K(x,x')$ 
 and noting that, making the special evaluation at $[x] -[x']$, this is an expansion in which 
 only  hook partitions appear, for which the Schur functions can be expressed as finite bilinear sums over
  the complete symmetric functions $\{h_i\}$ \cite{Mac}. 
  \end{proof}

\subsection{The multicurrent correlator as generating function}
\label{multicurrent_correl}

In terms of the indeterminates $(x_1, \dots, x_n)$, we define the following correlators
\bea
W_{n}({\bf s}; x_1,\dots,x_n )  := \left.\left( \Big(\prod_{i=1}^n \nabla(x_i) \Big)
 \tau^{(G, \beta, \gamma)} ({\bf t},\beta^{-1} {\bf s}) \right) \right\vert_{{\bf t}= {\bf 0}},
 \label{W_G_def}\\
 \tilde{W}_{n}({\bf s}; x_1,\dots,x_n )  := \left.\left( \Big(\prod_{i=1}^n \nabla(x_i) \Big)
\ln\,  \tau^{(G, \beta, \gamma)} ({\bf t},\beta^{-1} {\bf s}) \right) \right\vert_{{\bf t}= {\bf 0}},
\eea
where
\be
\nabla(x) := \sum_{i=1}^\infty x^{i-1} {\partial \over \partial t_i}.
\ee
\br
As will be seen in the fermionic representation in Sec.~\ref{fermionic}  below, these may be viewed as  {\em multicurrent correlators} and  {\em connected multicurrent correlators}.
\er

For $g \geq 0$ we denote the genus-$g$ component of the function $\tilde W_{n}({\bf s}; x_1,\dots,x_n )$
as follows
\bea
\tilde{W}_{g,n}({\bf s}; x_1,\dots,x_n )  := [\beta^{2g-2+n}] \tilde{W}_{n}({\bf s}; x_1,\dots,x_n ).
\eea

We now introduce another type of generating function for weighted Hurwitz numbers, 
both for  the connected and nonconnected case.
\begin{definition}
\bea
F_n({\bf s}; x_1,\dots,x_n) &\&:=  \sum_{\mu,\nu,\,\ell(\mu)=n}  \sum_d \gamma^{|\mu|} \beta^{d-\ell(\nu)} H_{G}^d(\mu,\nu)\, 
|\aut(\mu)| m_\mu(x_1,\dots,x_n) p_\nu({\bf s}) 
\label{F_n_expan}\\
\tilde F_n({\bf s};x_1,\dots,x_n) &\& :=  \sum_{\mu,\nu,\,\ell(\mu)=n}  \sum_d \gamma^{|\mu|} \beta^{d-\ell(\nu)} \tilde H^d_{G}(\mu,\nu)\,
|\aut(\mu)|  m_\mu(x_1,\dots,x_n) p_\nu({\bf s}) 
\label{tilde_F_n_expan}\\
\tilde F_{g,n}({\bf s}; x_1, \dots , x_n)  &\&:=  \sum_{\mu,\nu,\,\ell(\mu)=n}   \gamma^{|\mu|}  \tilde H^d_{G}(\mu,\nu)\, 
|\aut(\mu)| m_\mu(x_1,\dots,x_n) p_\nu({\bf s})\,,
\label{tilde_F_g_n_expan}
\eea
where $m_\mu(x_1, \dots, x_n)$ denotes the monomial sum symmetric function corresponding to the partition $\mu$,
and $(g, d)$ are related by the Riemann-Hurwitz formula
\be
2-2g = \ell(\mu) + \ell(\nu) - d.
\label{Riem_Hurw_g_d}
\ee
\end{definition}
Note that $F_n$, $\tilde F_n$ belong to $\mathbb{K}[x_1,\dots,x_n; {\bf s}; \beta, \beta^{-1}][[\gamma]]$ 
and $\tilde F_{g,n}$ to $\mathbb{K}[x_1,\dots,x_n; {\bf s}][[\gamma]]$.
These are related to the multicurrent correlators as follows
\begin{theorem}
\bea
\label{firW}
W_{n}({\bf s}; x_1,\dots,x_n ) &\&= {\partial^n\over \partial x_1 \cdots \partial x_n}F_{n}({\bf s}; x_1,\dots,x_n), 
\label{W_n_F_n}\\
\tilde W_{n}({\bf s}; x_1,\dots,x_n) &\&={\partial^n\over \partial x_1 \cdots \partial x_n}\tilde F_{n}({\bf s}; x_1,\dots,x_n), 
\label{tilde_W_n_F_n}\\
\tilde W_{g,n}({\bf s}; x_1,\dots,x_n) &\&= {\partial^n\over \partial x_1 \cdots \partial x_n}\tilde F_{g,n}({\bf s}; x_1,\dots, x_n). 
\label{tilde_W_g_n_F_n}\
\eea
\end{theorem}
\begin{proof}
Applying the operator   $\prod_{i=1}^n \nabla(x_i) $ to the symmetric function $p_\mu({\bf t})$ gives
\bea
\left.\left(\prod_{i=1}^n \nabla(x_i)\right) p_\mu({\bf t}) \right\vert_{{\bf t}= {\bf 0}}&\& 
=\left. \prod_{i=1}^n \left(\sum_{j_i =1}^\infty x_i^{j_i -1} {\partial \over \partial t_{j_i}}\right)
\left(\prod_{k=1}^{\ell(\mu)} \mu_k t_{\mu_k}\right)\right\vert_{{\bf t}= {\bf 0}} \cr
&\& =\delta_{\ell(\mu), n} \sum_{\sigma \in S_n} \prod_{k=1}^n \mu_k x_{\sigma(k)}^{\mu_k -1} 
 =\delta_{\ell(\mu), n}|\aut(\mu)|  {\partial^n m_\mu (x_1, \dots, x_n) \over \partial x_1 \cdots \partial x_n} 
\eea
by (\ref{m_lambda}).
Hence, applying it to $\tau^{(G, \beta, \gamma)} ({\bf t},\beta^{-1} {\bf s})$, using the expansion (\ref{tau_G_H}), we obtain
from the definition (\ref{W_G_def}) of $W_n({\bf s}, x_1, \dots, x_n)$
\bea
W_n({\bf s}, x_1, \dots, x_n) &\&= \sum_{\mu,\nu,\,\ell(\mu)=n}  \sum_d \gamma^{|\mu|} \beta^{d-\ell(\nu)}  H^d_{G}(\mu,\nu)
|\aut(\mu)|  {\partial^n m_\mu (x_1, \dots, x_n) \over \partial x_1 \cdots \partial x_n}  p_\nu({\bf s})  \cr
 &\&=  {\partial^n\over \partial x_1 \cdots \partial x_n}F_{n}({\bf s}; x_1,\dots,x_n) 
\eea
by (\ref{F_n_expan}), proving  (\ref{W_n_F_n}). The same calculation proves  the connected case   (\ref{tilde_W_n_F_n}),
and (\ref{tilde_W_g_n_F_n}) follows from equating like powers of $\beta$ in the series expansion.
\end{proof}
Henceforth, we suppress the explicit ${\bf s}$ dependence in the arguments and denote these quantities 
simply as $W_{n}(x_1,\dots,x_n ) $, etc. It follows from  (\ref{W_G_def}) that $W_{n}(x_1,\dots,x_n)$ may 
equivalently be expressed directly in terms of derivatives of special evaluations of 
the $\tau$-function $\tau^{(G, \beta, \gamma)}({\bf t}, {\bf s})$.

\begin{lemma}\label{prop:Wndeterminant} 
\be 
W_n(x_1,\dots,x_n) = 
[\epsilon_1, \dots , \epsilon_n]\left( 
 K_n(x_1+\epsilon_1,\dots, x_n+\epsilon_n; x_1,\dots,x_n)
\big/\det\left(\tfrac{1}{x_i-x_j+\epsilon_i}\right)  \right),
 \label{eq:Wndeterminant}
\ee
where both sides are viewed as elements of $\mathbb K[x_1,\dots,x_n,{\bf s},\beta,\beta^{-1}][[\gamma]]$.
\end{lemma}

For connected functions, 
we get the specially elegant relations:
\begin{proposition}\label{prop:detConnected}
\be\label{eq:detConnected1}
\tilde W_1(x) = \lim_{x'\to x} \left( K(x,x') - \frac{1}{x-x'}\right),
\ee
\be\label{eq:detConnected2}
\tilde W_2(x_1,x_2) = \left( -K(x_1,x_2)K(x_2,x_1)  - \frac{1}{(x_1-x_2)^2}\right),
\ee
and for $n\ge 3$
\be\label{eq:detConnectedn}
\tilde W_n(x_1,\dots,x_n) =
\sum_{\sigma \in \mathfrak S_n^{\rm 1-cycle}}
(-1)^\sigma \prod_{i} K(x_i,x_{\sigma(i)}),
\ee
where the last sum is over all permutations in $\mathfrak{S}_n$ having only one cycle.
These equalities hold in 
$\mathbb K(x,x_1,\dots,x_n)[{\bf s},\beta,\beta^{-1}][[\gamma]]$.
\end{proposition}
A detailed proof of this result is provided in the appendix to \cite{ACEH2}.


\subsection{The Christoffel-Darboux relation }

Define the functions $\xi(x,{\bf s})$ and $S(x, {\bf s})$ by
\bea
\xi(x,{\bf s}):=\sum_{k=1}^\infty s_k x^k, \\
\label{xidef}
S(x,  {\bf s}):=x \frac{d}{d x}\xi(x,{\bf s})=\sum_{k=1}^\infty k s_k x^k.
\label{Sdef}
\eea
In the following, we assume that only a finite number of variables $s_k$ are nonzero, 
so $\xi(x, {\bf s})$ and $S(x,  {\bf s})$ are polynomials
\be
\xi(x,{\bf s})=\sum_{k=1}^L s_k x^k, \quad S(x, {\bf s})=\sum_{k=1}^L k s_kx^k
\ee 
of degree $L$. We also assume that only the first $M$ parameters $\{c_i\} $ are nonzero, 
so the weight generating function $G(z)$  is also a polynomial,  of degree $M$ 
\be
 G(z) = \prod_{i=1}^M(1+c_i z). 
\ee
For brevity, the explicit dependence on the parameters ${\bf s}$ will henceforth also be suppressed and $S(x,  {\bf s})$ 
and  $\xi(x,{\bf s})$ just denoted $S(x)$,  $\xi(x)$.

\begin{definition}
Define the operators
\be
\Delta_\pm(x):=\pm\beta e^{\mp\beta^{-1}\xi(x)} \circ D \circ e^{\pm\beta^{-1}\xi(x)}=S(x)\pm \beta D
\label{eq:defDeltapm}
\ee
and
\be
V_\pm(x):=\gamma^{-1}x^{-1}e^{\mp\beta^{-1}\xi(x)}\circ R_{\pm}\circ e^{\pm\beta^{-1}\xi(x)}=G(\Delta_\pm(x)).
\label{eq:defVpm}
\ee
\end{definition}

In the following, we  use boldface uppercase letters, such as $\mat{A}$, $\mat{F}$ or $\mat{V}(x)$ to denote finite dimensional matrices.
\begin{definition}
Define the polynomial $A(r,t)$ of degree $L M-1$ in each variable $(r,t)$   which determines the
	$LM \times LM$ nonsingular matrix $\mat{A}=(\mat{A}_{ij})_{0\leq i,j<LM}$  as follows
\be
	A(r,t):=\left(r\, V_-(t)-t\, V_+(r)\right)\left( \frac{1} {r-t}\right) = \sum_{i=0}^{LM-1} \sum_{j=0}^{LM-1} \mat{A}_{ij}r^i t^j.
\label{eq:defA}
\ee
\end{definition}

The highest total degree term is
\be
g_M \, (L s_L)^M \frac{r\, t^{LM}-t\, r^{LM}}{r-t}=-g_M \, (L s_L)^M \sum_{j=1}^{LM-1} r^j t^{LM-j}
\ee
so 
\be
\det\, \mat{A} =(-1)^\frac{LM(LM-1)}{2}\, g_M^{LM-1} \, (L s_L)^{M(LM-1)}.
\ee

The following {\it Christoffel-Darboux} type relation, expressing $K(x,x')$ as a finite sum of basis elements is proved in \cite{ACEH1, ACEH2}. 
\begin{theorem}
\label{thm:CD} 
\be
K(x,x')=\frac{1}{x-x'} A(R_+,R_-')\Psi^+_0(x)\Psi^-_0(x').
\label{eq:CD}
\ee
\end{theorem}

\br
In the above, we  made the assumption that $G$ is a polynomial, of degree $M$. We  keep this assumption in the following; however, because any given coefficient in the $\tau$-function is a polynomial in the coefficients of $G$, it is possible to drop it in many instances by taking projective limits. Indeed if $G(z)$ is an infinite formal power series and $G_M(z)$ is its degree-$M$ truncation, then all the correlators corresponding to $G(z)$ are the limit (in the sense of convergence of individual coefficients) of those corresponding to $G_M(z)$ as $M\rightarrow \infty$. This enables us, for example, to include 
the exponential weight generating function $G(z)=e^z$ which gives the case of simple (single and double) Hurwitz numbers
studied in \cite{Ok, Pa}.
\er

\section{The spectral curve}
\label{spec_curve}

\subsection{Quantum spectral curve}
\label{quantum_sc}

Applying the operators $\beta D$ and $S(R_\pm)$ to the series expansions  (\ref{Psi_zero}), (\ref{Phi_zero})
we obtain:
\begin{theorem}
The function $\Psi^+_0(x)$  satisfies  
\be
\left(\beta D-S(R_+)\right)\Psi_0^+(x)=0,
\label{eq:quantumCurve}
\ee
and $\Psi^-_0(x)$ satisfies the dual equation
\be
\left(\beta D+S(R_-)\right)\Psi_0^-(x)=0.
\label{eq:quantumCurveDual}
\ee
\end{theorem}

The linear ordinary differential operators  $\left(\beta D-S(R_+)\right)$ and $\left(\beta D+S(R_-)\right)$ with polynomial coefficients
that annihilate $\Psi^+_0(x)$, $\Psi^-_0(x)$  in   (\ref{eq:quantumCurve}), (\ref{eq:quantumCurveDual})
will be referred to as the  {\it quantum spectral curve} and its dual.

\subsection{Classical spectral curve}
\label{classical_sc}

The classical spectral curve is obtained by replacing the Euler operator appearing
in (\ref{eq:quantumCurve}), (\ref{eq:quantumCurveDual}) and (\ref{Rdef}) by the 
product $\beta^{-1} x y$, or by $-\beta^{-1} x y$ in the dual case.
Here $y$ is viewed as the classical variable canonically conjugate to $x$
(and $\beta$ is identified with Planck's constant $\hbar$). This gives
\be
x y = S(\gamma x G(x y))
\ee
as the classical spectral curve. This is a rational curve 
realized as an $LM$-sheeted branched cover of $\Cb \Pb^1$.
The significance of this is that  $y \, dx$, viewed as a 
meromorphic form on this Riemann surface, is equal to
the  rational differential $\tilde{W}_{0, 1}(x) dx$, as shown in ref.~\cite{ACEH2}.
\begin{proposition} 
\label{W_tilde_0}
\be
\tilde{F}_{0, 1}(x) = \int y\, dx.
\ee
\end{proposition}
The spectral curve  admits a rational parametrization by the  functions $X(z)$ and $Y(z)$ defined as:
\bea\label{zcurve}
	X(z)&:=&\frac{z}{\gamma G(S(z))} , \cr
	Y(z)&:=& \frac{S(z)}{z}\,\gamma G(S(z)).
\eea
Note  that we have
\be\
X(z) Y(z)=S(z).
\label{xyrel}
\ee

The ramification points under the projection to the Riemann sphere $(x,y) \ra x$  are given by 
\be
{d X \over dz} =0, \quad \text{with }\  {dY \over dz} \ne 0.
\ee
Equivalently, the ramification points $(X(z), Y(z))$ are given by $z$-values in the set $\mathcal{A}$ of
roots of the polynomial defined as follows
 \be
\sigma(z) :=G(S(z)) - z S'(z) G'(S(z)),
\ee
and the point $z=\infty$, which we do not include in $\mathcal{A}$.

Assuming  $x$ is {\it not} a branch point and picking a point $(X(z), Y(z))$
over $x$ with uniformizing parameter $z$, 
over  a neighbourhood of $x = X(z)$ that contains no branch
points, the equation $X(z)=x$, has $LM$ solutions, which we  label as follows:  
\be \label{eq:defZordering}
\{z^{(0)}(z):=z, \ z^{(1)}(z), \dots, \ z^{(LM-1)}(z)\}
\ee
so
\be 
X(z^{(i)}(z))=x, \quad i=0, \dots, LM-1,
\ee
The ordering for the roots $\{z^{(1)}(z), \dots, z^{(LM-1)}(z)\}$ is fixed arbitrarily, once and for all, within this neighbourhood;
and the functions $z^{(i)}(z)$ are analytic in $z$, as are $X(z)$, $Y(z)$ and $1/X'(z)$.

\begin{definition}
For $z$ as above let $\Vb(z)$ denote the Vandermonde matrix with elements
\be
\Vb(z)_{i j} = \left(z^{(j)}(z)\right)^{i-1}.
\ee
\end{definition}

\section{Linear first order system for  the adapted basis}
\label{lin_diff_system}

\subsection{The constant coefficient infinite system}
 \label{inf_system}
 
 Define two doubly infinite column vectors whose components are the functions $\Psi^+_k(x)$ and $\Psi^-_k(x)$.
\be
\vec{{\Psi}}^+_{\infty}:=\begin{pmatrix}
\vdots\\
{\Psi}^+_{-1}\\
{\Psi}^+_{0}\\
{\Psi}^+_{1}\\
\vdots\\
\end{pmatrix},\,\,\,\,\,\,\,\,\,\,\,\,\,
\vec{{\Psi}}^-_\infty:=\begin{pmatrix}
\vdots\\
{\Psi}^-_{-1}\\
{\Psi}^-_{0}\\
{\Psi}^-_{1}\\
\vdots\\
\end{pmatrix}.
\ee 

Here, and in what follows, the word \emph{constant} refers to a quantity that does not depend on $x$ or $x'$. 
Now define four doubly infinite constant matrices $ Q^{\pm}$ and $P^{\pm}$ 
by their generating functions
\bea
{\mathcal P}^{\pm}(t,r)=\sum_{i,j=-\infty}^\infty P_{ij}^{\pm}t^{i}r^{j}=\Delta_\pm(r)\sum_{i=-\infty}^{\infty} (tr)^i,\\
{\mathcal Q}^{\pm}(t,r)=\sum_{i,j=-\infty}^\infty Q_{ij}^{\pm}t^{i}r^{j}=V_\pm(r) r^{-1} \sum_{i=-\infty}^{\infty} (tr)^i.
\eea
 Note that $P^{\pm}$ are upper triangular, with diagonal entries $\PP^{\pm}_{kk}=\pm k \beta$ and upper triangular
 entries $\PP^{\pm}_{kj} = (j-k) s_{j-k}$, while $Q^{\pm}$ are finite band matrices, with just one band below the principal
 diagonal, and $LM-2$ above it.

As shown in refs.~\cite{ACEH1,  ACEH2}, the recursion relations (\ref{recursions_Psi_plus_k}), (\ref{recursions_Psi_minus_k}),  together with the 
commutator relation (\ref{Rcom}) can be expressed equivalently as the following linear, constant
coefficient system for $\Psi^+_k(x)$ and $\Psi^-_k(x)$  :
\begin{theorem} 
\bea
\label{eq:P+}
\beta D \vec{{\Psi}}^+_\infty= P^{+}\vec{{\Psi}}^+_\infty,
\\
\frac{1}{\gamma x}\vec{{\Psi}}^+_\infty= Q^{+}\vec{{\Psi}}^+_\infty,
\label{eq:Q+}
\eea
and
\bea\label{eq:P-}
-\beta D \vec{{\Psi}}^-_\infty= P^{-}\vec{{\Psi}}^-_\infty,
\\
\label{eq:Q-}\frac{1}{\gamma x}\vec{{\Psi}}^-_\infty= Q^{-}\vec{{\Psi}}^-_\infty.
\eea
\end{theorem}

\subsection{Finite dimensional first order system: folding}
\label{dual_cov_deriv}

Using the mulitplicative recursion relations   (\ref{eq:Q+}), (\ref{eq:Q-}) and the finite band structure 
of the matrices $\QQ^{\pm}$, we can ``fold'' the constant coefficient terms on the RHS of 
eqs. (\ref{eq:P+}) and (\ref{eq:P-}) onto a  finite ``window'' consisting of the first
$LM$ components of the vectors $\vec{{\Psi}}_{\infty}^+$ and $\vec{{\Psi}}_{\infty}^-$,
at the cost of introducing a polynomial dependence upon $1/x$ in the coefficients.
In fact, due to the special structure of eqs.~(\ref{eq:Q+}), (\ref{eq:Q-}),
 this dependence turns out to just be linear in $1/x$.

Define the two column vectors $\vec{{\Psi}}^+(x)$ and $\vec{{\Psi}}^-(x)$ of dimension $LM$ by 
\be
\vec{{\Psi}}^+:=\begin{pmatrix}
{\Psi}^+_{0}\\
{\Psi}^+_{1}\\
{\Psi}^+_{2}\\
\dots\\
{\Psi}^+_{LM-1}\\
\end{pmatrix},\quad \quad 
\vec{{\Psi}}^-:=\begin{pmatrix}
{\Psi}^-_{0}\\
{\Psi}^-_{1}\\
{\Psi}^-_{2}\\
\dots\\
{\Psi}^-_{LM-1}\\
\end{pmatrix}.
\ee 
Let the elements $\tilde{\mat{E}}_{ij}(x)$ of the $ML \times ML$ matrix $\tilde{\mat{E}}(x)$ be defined by
the generating function expression
\be
\sum_{i,j=0}^{LM-1}\tilde{\mat{E}}(x)_{ij}r^{i}t^{j} := \left(\Delta_+(r)rV_-(t)-\Delta_-(t)tV_+(r)\right)\left(\frac{1}{r-t}\right) -\frac{1}{\gamma x} rt \frac{S(r)-S(t)}{r-t},
\label{tildeEE}
\ee
and define a further pair of $LM \times LM$ matrices $\mat{E}(x)$ and $\mat{E}'(x)$ by
\bea
\mat{E}(x) &\& := (\mat{A}^T)^{-1} \tilde{\mat{E}}^T
\label{def_EE}\\
\mat{E}'(x) &\& := \mat{A}^{-1} \tilde{\mat{E}}
\label{def_EE'}
\eea
where $\mat{A}$ is the $LM \times LM$ matrix entering in the Christoffel-Darboux relation in Theorem~\ref{thm:CD}.
The matrices $\mat{E}(x)$ and $\mat{E}'(x)$ thus satisfy the following duality relation:
\be
\label{maste}
\mat{A}{\mat{E}}(x)-{\mat{E}}'(x)^{T}\mat{A}=0.
\ee
The {\it folded} version of the linear relations (\ref{eq:P+}) and (\ref{eq:P-}) is then \cite{ACEH2}:
\begin{theorem}
\label{prop:finiteSystem}

The following finite dimensional differential systems follow from (\ref{eq:P+}) and (\ref{eq:P-})
\bea\label{Edef}
\beta D \vec{{\Psi}}^+&=&{\mat{E}}(x)\vec{{\Psi}}^+,\\
-\beta D \vec{{\Psi}}^-&=&{\mat{E}'}(x)\vec{{\Psi}}^-.
\eea
\end{theorem}

\subsection{Adjoint differential system}

\begin{definition}
Let $\mat{M}(x)$ be the rank-$1$, $LM \times LM$ matrix  defined by
\be
\label{defM}
	\mat{M}(x) := \vec{\Psi}^-(x) \vec{\Psi}^+(x)^T \, \mat{A}.
\ee
with entries 
\be\label{eq:entryM}
	\mat{M}(x)_{ij} = \Psi^-_i(x) \sum_{k=0}^{LM-1} \Psi^+_k(x) \mat{A}_{k,j} 
\ee
viewed as are elements of $\mathbb K[x,s,\beta,\beta^{-1}][[\gamma]]$.
\end{definition}

The matrix $\mat{M}(x)$ has the following properties:
\begin{proposition}\label{prop:Mpositive}
The entries of $\mat{M}(x)$ are elements of $\mathbb K[x,s,\beta][[\gamma]]$, \textit{i.e.} they contain no negative power of $\beta$.
Moreover,  $\mat{M}(x)$ is a rank 1 projector:
\be
\mat{M}(x)^2=\mat{M}(x)
\quad , \quad
\Tr \,\mat{M}(x)=1.
\ee
and satisfies the adjoint differential system
\be\label{eq:adjODE}
\beta x \frac{d}{dx} \mat{M}(x) = [\mat{E}(x),\mat{M}(x)].
\ee
\end{proposition}

\subsection{$\beta$-expansion of $M$}

Let $\mat{F}$ be the elementary  $LM\times LM$ matrix
\be
\mat{F}:={\rm diag}(1,0,\dots,0).
\ee
We then have \cite{ACEH2}:
\begin{theorem}\label{defthm:MWKB}~
	
	(I) There exists a unique $LM\times LM$ matrix $\mat{M}^{\rm WKB}(z)$ formal series expansion with the following properties: 
	$\mat{M}^{WKB}(z)$ is a rank 1 projector, its entries belong to $\mathbb K[\mathbf s][[z,\beta]]$, and $\mat{M}^{\rm WKB}(z)$ is a solution of the ODE
\be
\label{eq:adjODEz}
		\beta \frac{z G(S(z))}{\sigma(z)}  \frac{d}{dz} \mat{M}^{WKB}(z) = [\mat{E}(X(z)),\mat{M}^{WKB}(z)]
\ee
with leading order:
\be
	\mat{M}^{\rm WKB} (z) = \mat{V}(z) \mat{F} \mat{V}(z)^{-1} + O(\beta).
\ee
	Moreover, for each $k\geq 0$, all entries of the matrix $\mat{M}^{(k)}(z)=[\beta^k] \mat{M}^{\rm WKB}(z)$ 
	are rational functions of $z$, with poles only at the ramification points (elements of $\mathcal{A}$) or at $z=\infty$.
	
	(II) For each $k\geq 0$, the following equality holds as formal power series of $z$:
	\be\label{eq:MzEqualsMx}
	\mat{M}^{(k)}(z) = [\beta^k] \mat{M}(X(z)).
	\ee

\end{theorem}

\subsection{The current correlators $\tilde W_n$}

It follows  from the determinantal formula (\ref{prop:detConnected}) and the Christoffel-Darboux Theorem \ref{thm:CD} 
 that \cite{ACEH2}:
\begin{proposition} \label{prop:WgntoM}
\be
\tilde W_1(x) 
	= \frac{1}{\beta x} \Tr  \mat{M}(x) \mat{E}(x),
\ee
\be
\tilde W_2(x_1,x_2) 
= \frac{\Tr  \mat{M}(x_1)\mat{M}(x_2) }{(x_1-x_2)^2} - \frac{1}{(x_1-x_2)^2},
\ee
and for $n\geq 3$
\be
\tilde W_n(x_1,\dots,x_n) 
= \sum_{\sigma \in \mathfrak S_n^{\rm 1-cycle}}
(-1)^\sigma \frac{\Tr \left( \prod_{i} \mat{M}(x_{\sigma(i)}) \right) }{\prod_i (x_{\sigma(i)} - x_{\sigma(i+1))})},
\ee
\end{proposition}

It also follows that

\begin{proposition}\label{prop:WgntoMWKB}
For $n\ge 0$, let $z_1,z_2,\dots z_n$ be $n$ generic distinct points on the Riemann sphere, the generating functions $\tilde W_{g,n}$ evaluated at $x_i=X(z_i)$,  are rational fractions of $z_1,\dots,z_n$:
\be
\sum_g \beta^{2g-2+n} \tilde W_{g,n}(X(z_1),\dots,X(z_n)) 
= \sum_{\sigma \in \mathfrak S_n^{\rm 1-cycle}}
	(-1)^\sigma \frac{\Tr \left( \prod_{i} \mat{M}^{\rm WKB}(z_{\sigma(i)}) \right) }{\prod_i (X(z_{\sigma(i)}) - X(z_{\sigma(i+1))}))},
\ee
\end{proposition}

\section{Fermionic representations}
\label{fermionic}

In the following, we use the standard free fermion representation, in which the creation and annihilation operators
$(\psi_i, \psi^*_i)_{i\in \Zb}$ satisfy the anticommutation relations
\be
[\psi_i, \psi_j]_+ =0, \quad [\psi^*_i, \psi^*_j]_+ =0, \quad [\psi_i, \psi^*_j]_+ =\delta_{ij}, 
\ee
and the vacuum state $|0\rangle$ is annihilated by the positive index annihilation operators
and the negative creation operators
\be
\psi_{-i}|0 \rangle = 0, \quad  i= 1, \dots ,\infty, \quad \psi^*_{i}|0 \rangle = 0, \quad  i= 0, \dots ,\infty.
\ee
The Fermi field and its adjoint are denoted
\be 
\psi(z) = \sum_{i=-\infty}^\infty \psi_i  z^i, \quad \psi^*(z) = \sum_{i=-\infty}^\infty \psi^*_i  z^{-i-1}, 
\ee
and the current operator components are denoted
\be
J_i := \sum_{j= -\infty}^{\infty} \no{\psi_j \psi^*_{i+j}}  \quad i\in \Zb.
\ee
The abelian groups $\Gamma_{\pm}$ of {\it shift flows}, are defined as
\be
\Gamma_+ := \{ \hat{\gamma}_+({\bf t})\}, \quad \Gamma_- := \{ \hat{\gamma}_-({\bf s})\}
\ee
where
\be
\hat{\gamma}_+({\bf t}) := e^{\sum_{i=1}^\infty t_i J_i}, \quad \hat{\gamma}_-({\bf s}) := e^{\sum_{i=1}^\infty s_i J_{-i}}.
\ee

Another infinite abelian group that enters in the fermionic representation of our $\tau$-functions
is the group ${\bf C}$ of (generalized) convolution flows \cite{HO}, consisting of diagonal elements
\be
{\bf C} := \{ \hat{C}_\rho := e^{\sum_{i=-\infty}^\infty T_i \no{\psi_i \psi^*_i }} \},
\ee
where
\be
\rho_i = e^{T_i},
\ee
may be viewed as the Fourier coefficients of a function (or distribution) on the unit circle $\{z : |z|=1\}$,
\be
\rho(z) = \sum_{i=-\infty}^\infty \rho_i z^{-i-1}.
\ee

The next three subsections \ref{tau_fermionic_sec}, \ref{pair_correl_fermionic} and \ref{current_correl_fermionic} give
fermionic representations of all the quantities appearing above: the $\tau$-function, the Baker function, the adapted bases, the pair correlators, the multicurrent correlator. For proofs and further details of the relations between them, the reader should consult the companion paper \cite{ACEH1}.
\subsection{The $\tau$-function,  Baker function and adapted bases }
\label{tau_fermionic_sec}
The $\tau$-function (\ref{tau_G_beta_gamma}) has the following representation as a 
fermonic vacuum state expectation value (VEV)
\be
\tau^{(G, \beta, \gamma)} ({\bf t}, {\bf s}) = \langle 0 | \hat{\gamma}_+({\bf t}) \hat{C}_\rho \hat{\gamma}_-({\bf s}) | 0 \rangle,
\label{tau_fermionic}
\ee
where the coefficients $\rho_i$ are given by (\ref{rho_j+}), (\ref{rho_j-}).

Viewing this as a parametric family of KP $\tau$-functions (where  ${\bf s}=(s_1, s_2, \dots)$ are just viewed as additional parameters), the corresponding Baker function and its dual have the fermionic representation  \cite{ACEH1}
\bea
\Psi^-_{(G,\beta,\gamma)}(\zeta, {\bf t}, {\bf s}) =  { \langle 0 | \psi^*_0 \hat{\gamma}_+({\bf t}) \psi(\zeta) \hat{C}_\rho \hat{\gamma}_-({\bf s})| 0 \rangle
\over \tau^{(G, \beta, \gamma)} ({\bf t}, {\bf s}) },\\
\Psi^+_{(G,\beta,\gamma)}(\zeta, {\bf t}, {\bf s}) = { \langle 0 | \psi_{-1} \hat{\gamma}_+({\bf t}) \psi^*(\zeta) \hat{C}_\rho \hat{\gamma}_-({\bf s})| 0 \rangle
\over \tau^{(G, \beta, \gamma)} ({\bf t}, {\bf s}) }.
\eea

	The adapted basis  defined by \eqref{Psi_plus_k_ser}  has the  fermionic representation  \cite{ACEH1}
	\be
\Psi_k^+(x)
  :=\begin{cases} \gamma \langle 0 |\psi^*(1/x) \hat{C}_\rho \hat{\gamma}_-(\beta^{-1}{\bf s}) \psi_{k-1} | 0\rangle  \quad \text{if} \ k \ge 1, \cr
\gamma \langle 0 | \psi_{k-1} \hat{\gamma}^{-1}_-(\beta^{-1}{\bf s})\ \hat{C}^{-1}_\rho \psi^*(1/x) \hat{C}_\rho \hat{\gamma}_-(\beta^{-1}{\bf s})| 0\rangle \quad  \text{if} \ k\le 0,
\end{cases} 
\label{wk_rho_s}
\ee
and the dual basis  defined by \eqref{Psi_minus_k_ser}  is given  \cite{ACEH1} by
\be
\Psi_k^-(x):=  \begin{cases} \langle 0 | \psi(1/x) \hat{C}_\rho \hat{\gamma}_-(\beta^{-1}{\bf s})\psi^*_{-k} |0 \rangle  \quad  \text{if} \ k \ge 1 \cr
\langle 0 |\psi^*_{-k}\hat{\gamma}^{-1}_-(\beta^{-1}{\bf s})\ \hat{C}^{-1}_\rho \psi(1/x)  \hat{C}_\rho
 \hat{\gamma}_-(\beta^{-1}{\bf s})  |0 \rangle \quad  \text{if} \ k\le 0 .
\end{cases} \\
\label{wk_rho_s*}
\ee

\subsection{The $n$-pair correlator $K_n$ }
\label{pair_correl_fermionic}
The $n$-pair correlator (\ref{K_n_tau_def}) has the following fermionic representation \cite{ACEH1}
\be
K_n(x_1,\dots,x_n; x'_1,\dots,x'_n) = 
{1\over \prod_{i=1}^n x_i x'_i}
 \langle 0 | \prod_{i=1}^n \psi({1/ x'_i}) \psi^*({1/ x_i}) \hat{C}_\rho \hat{\gamma}_-(\beta^{-1}{\bf s}) | 0 \rangle.
\ee

\subsection{The multicurrent  correlator $W_n$}
\label{current_correl_fermionic}

It follows from the fermonic representation (\ref{tau_fermionic}) of the $\tau$-function $\tau^{(G, \beta, \gamma)} ({\bf t}, {\bf s})$
and the definition (\ref{W_G_def}) of the multicurrent correlator  $W_n(x_1, \dots, x_n)$ that
this may be expressed as the following fermionic vacuum expectation value \cite{ACEH1}
\be
W_n(x_1, \dots, x_n) = {1\over \prod_{i=1}^n x_i}\langle 0 | \prod_{i=1}^n J_+(x_i) \hat{C}_\rho \hat{\gamma}_-(\beta^{-1} {\bf s}) | 0 \rangle,
\ee
where
\be
J_+(x) := \sum_{i=1}^\infty x^i J_i.
\ee

\section{Topological recursion}
\label{top_rec}

\subsection{Definitions of $\tilde{\omega}_{g,n}$ and ${\mathcal W}_{g,n}$}
\label{omega_gn}

The key invariants computed by  topological recursion in this case are the forms $\tilde\omega_{g,n}$ defined in terms of the correlators $\tilde W_{g,n}$ as follows:
\begin{definition}
Let
\begin{multline}
\tilde \omega_{g,n}(z_1,\dots,z_n) := \tilde W_{g,n}(X(z_1),\dots,X(z_n)) \,X'(z_1)\dots X'(z_n) dz_1 \dots dz_n \cr
+ \delta_{g,0}\delta_{n,2} \ \frac{X'(z_1) X'(z_2) \ dz_1 dz_2}{(X(z_1)-X(z_2))^2} \ .
\end{multline}
\end{definition}
It follows \cite{ACEH2} that $\tilde \omega_{g,n}$ for the stable cases are rational functions of all the $z_i$'s, with poles only at elements of $\mathcal{A}$ (ramification
 points of the spectral curve) or at infinity. In fact we have the stronger property (see \cite{ACEH2} for a detailed proof):
\begin{proposition}
When $2g-2+n>0$, $\tilde\omega_{g,n}(z_1,\dots,z_n)$ has poles only at ramification points, and no poles at $z_i=\infty$.
\end{proposition}

In order to state our topological recursion in a more compact manner, it is convenient to introduce the following auxiliary quantity.
\begin{definition}
\begin{multline}
{\mathcal W}_{g,n}(z,z',z_2,\dots,z_n)
:= \tilde \omega_{g-1,n+1}(z,z',z_2,\dots,z_n) \\
+ \sum'_{g_1+g_2=g,\, I_1\uplus I_2=\{z_2,\dots,z_n\}}
\tilde \omega_{g_1,1+|I_1|}(z,I_1)\tilde \omega_{g_2,1+|I_2|}(z',I_2), 
\end{multline}
where $\sum'$ means that we exclude the two terms $(g_1,I_1)=(0,\emptyset)$ and $(g_2,I_2)=(0,\emptyset)$.
\end{definition}

\subsection{Topological recursion for the $\tilde \omega_{g,n}$'s}
\label{top_rec_hurwitz}

As shown in \cite{ACEH2},  the $\tilde W_n$'s satisfy a set of equations called ``loop equations''  \cite{BEM, BBE}.
The unique formal $\beta$-power series solution of the loop equations, having poles only at the ramification points, is the one given by the topological recursion  relations (Theorem 4.5. of \cite{EO1}; see\cite{BEO} for another proof). 
\begin{theorem}\label{th:toprec}
The $\tilde\omega_{g,n}$'s satisfy the topological recursion relations:
	If all ramification points $a\in \mathcal{A}$ are simple, with local Galois involution $z\mapsto \sigma_a(z)$, we have for $(g,n)\neq (0,1),(0,2)$:
\begin{multline}
\tilde\omega_{g,n}(z_1,\dots,z_n) 
	= -\,\sum_{a\in \mathcal{A}} \Res_{z\to a} \Big[ \frac{dz_1}{z-z_1} - \frac{dz_1}{\sigma_a(z)-z_1}  \Big] \,\frac{{\mathcal W}_{g,n}(z,\sigma_a(z),z_2,\dots,z_n)\ }{2(Y(z)-Y(\sigma_a(z)))\,dX(z)} .
\end{multline}

If some ramification points have higher order, then the $\tilde\omega_{g,n}$s satisfy the higher order version of the topological recursion relation
defined in \cite{BouchE}. 

For $(g,n)=(0,2)$ we have
\begin{equation}
\tilde \omega_{0,2}(z_1,z_2) = \frac{dz_1 dz_2}{(z_1-z_2)^2} \ .
\end{equation}

\end{theorem}
A detailed proof of this main result is provided in ref. \cite{ACEH2}.


\section{Some consequences of topological recursion}

The first consequence of Theorem \ref{th:toprec} is that $\tilde \omega_{g,n}(z_1,\dots, z_n)$ can be computed explicitly for each $g$ and $n$, which in turn gives an explicit rational parametrization of the function $\tilde W_{g,n}(x_1,\dots,x_n)$.

Theorem \ref{th:toprec} has many other deep implications, of which a short glimpse is the following.

\begin{itemize}

\item{\bf Structure properties:}
When $2g-2+n>0$, the $\tilde F^{(G,\gamma)}_{g,n}(x_1,\dots,x_n)$ are not only formal series in $\gamma$, but are in fact algebraic functions of $\gamma$ and the $x_i$'s, and of all parameters ${\bf c}=(c_1, c_2, \dots)$ defining $G$, as well as the parameters ${\bf s}=\{s_1, s_2, \dots\}$.
The radius of convergence in $\gamma$ is given by the  ramification point nearest  the origin:
\begin{equation}
|\gamma x_i| < R
\quad , \quad
 R= \min_j \{ |X(a_j)|\,\,\, |\,\,X'(a_j)=0\}.
\end{equation}
The $\tilde F_{g,n}(X(z_1),\dots,X(z_n))$, with $2g-2+n>0$, are rational functions of the arguments $z_i$,  with poles only at ramification points, of order $\leq 3(2g-2+n)$. For $(g,n)=(0,1)$ explicit expressions are known involving logarithms, with the same radius of convergence, (See~\cite{EO1}.)

\item The first few $\tilde F_{g,n}$ are easy to compute using topological recursion. For instance 
\begin{eqnarray}
\tilde F_{0,3} 
&=& \frac{X(z_1)}{X'(z_1)S'(z_1) (z_1-z_2)(z_1-z_3)} +{\rm sym}\,(z_1\leftrightarrow z_2,z_3)\cr
	&& - \sum_{b,\,S'(b)=0} \frac{X(b)}{X'(b)S''(b)}\,\frac{1}{(z_1-b)(z_2-b)(z_3-b)}.
	\label{eq:F03}
\end{eqnarray}

(For the case $L=1$ and $M$ arbitrary, ~\eqref{eq:F03},  was conjectured by John Irving (private communication to GC) in 2013. He  also obtained the expression of $\tilde{W}_{0,2}$ via elementary combinatorial manipulations, and conjectured that all the $\tilde{W}_{0,n}$ are rational functions of the parameters $z_i$. We refer the reader to a paper in preparation \cite{Irving} for these conjectures (which are here proved, and in \cite{ACEH2}, in greater generality) and for the combinatorial approach to the computation of $\tilde{W}_{0,2}$. Other combinatorial approaches to the computation of generating functions for constellations
may be found in ~\cite{Ch} and the thesis~\cite{Fa}, which contains in particular an analog of the cut-and-join equation\cite[Theorem 4.2]{Fa}.)

\item Topological recursion implies a form-cycle duality formula for all deformations, for example
\be
\frac{\partial}{\partial s_k} \tilde F_{g,n}(x_1,\dots, x_n)
= \operatorname{Res}_{z'\to \infty} \, \tilde F_{g,n+1}(x_1,\dots, x_n,X(z')) k z'^{k-1}dz',
\ee
\begin{eqnarray}
\frac{\partial}{\partial g_k} \tilde F_{g,n}(x_1,\dots, x_n)
&=& \operatorname{Res}_{z'\to \infty}  \tilde F_{g,n+1}(x_1,\dots, x_n,X(z')) \, d\left(\frac{S'(z')}{G(z')}\right),\cr
\frac{\partial}{\partial c_i} \tilde F_{g,n}(x_1,\dots, x_n)
&=& \operatorname{Res}_{z'\to \infty}  \tilde F_{g,n+1} (x_1,\dots, x_n,X(z')) \, d\left(\frac{S'(z')z'}{1+c_i z'}\right) .
\end{eqnarray}

\item It has many other deep consequences, which are well beyond this short summary. 
 For example, topological recursion implies some ELSV--like formulas as in \cite{EO3}. 
 These will be developed in the  longer version of this  work \cite{ACEH2}.

\end{itemize}

 \bigskip
\noindent 
\small{ {\it Acknowledgements.} 
The work of A. Alexandrov was supported by IBS-R003-D1, by RFBR grants 14-01-00547 and 15-52-50041YaF, and by the European Research Council 
(QUASIFT grant agreement 677368). 
G.Chapuy acknowledges support from the European Research Council, grant ERC-2016-STG 716083 ``CombiTop'', from the Agence Nationale de la Recherche, grant ANR 12-JS02-001-01 ``Cartaplus'' and from the City of Paris, grant ``\'Emergences 2013, Combinatoire \`a Paris''. 
B. Eynard was supported by the ERC Starting Grant no. 335739 ``Quantum fields and knot homologies'' funded by the European Research Council under the European Union's 
Seventh Framework Programme,  and also partly supported by the ANR grant Quantact : ANR-16-CE40-0017.
The work of J. Harnad was partially supported by the Natural Sciences and Engineering Research Council of Canada (NSERC) and the Fonds de recherche du Qu\'ebec, 
Nature et technologies (FRQNT).  

Much of this work was  completed while G.C. was affiliated to the Centre de Recherches Math\'ematiques de Montr\'eal as part of the CNRS UMI program;
he wishes to thank the CRM for the very good working conditions provided.   B.E. also wishes to thank the Centre de recherches math\'ematiques, 
Montr\'eal, for the Aisenstadt Chair grant,  and the FQRNT grant from the Qu\'ebec government that partially supported this joint project.
A.A. and J.H. wish to thank the Institut des Hautes \'Etudes Scientifiques for their kind hospitality for extended periods in 2017 - 2018,
 when much of this work was completed. The authors would also like to thank the organizers of the ``Moduli spaces, integrable systems, and topological recursions''  
thematic semester program (June 2015 - Jan. 2016) at the CRM in Montr\'eal, where part of this work was realized 
and the organizers of the January - March, 2017  thematic semester ``Combinatorics and interactions''  
at the Institut Henri Poincar\'e, where they were participants during the completion of this work. 

G.C. also thanks John Irving for sharing his results and Wenjie Fang, whose PhD thesis contains several insights that have helped to inspire this work. 
 B.E. and A.A. thank Piotr Su\l kowski and S. Shadrin for useful discussions.}
 
 Work supported by the Natural Sciences and Engineering Research Council of Canada (NSERC), the 
Fonds de recherche du Qu\'ebec - Nature et technologies (FRQNT),  the European Research Council (ERC), 
and the Agence Nationale de la Recherche (ANR).


\newcommand{\arxiv}[1]{\href{http://arxiv.org/abs/#1}{arXiv:{#1}}}

\bigskip
\noindent


\begin{thebibliography}{99}

\bibitem{AEG}   A.~Alexandrov,
 ``Enumerative Geometry, Tau-Functions and Heisenberg-Virasoro Algebra,''
  {\em Commun.\ Math.\ Phys.\ }  {\bf 338}, 195-249 (2015).

\bibitem{ACEH1} A.~Alexandrov, G.~Chapuy, B.~Eynard and J.~Harnad,
  ``Fermionic approach to weighted Hurwitz numbers and topological recursion,''
  {\em Commun.\ Math.\ Phys.} {\bf 360} no.2,  777  (2018).
  
 \bibitem{ACEH2} A.~Alexandrov, G. Chapuy, B. Eynard and J. Harnad, ``Weighted Hurwitz numbers  and topological recursion'',
\arxiv{1806.09738}.

\bibitem{AM}
  A.~Alexandrov,
  ``Matrix Models for Random Partitions,''
  {\em Nucl.\ Phys.\ B} {\bf 851} 620  (2011).

\bibitem{ALS} 
  A.~Alexandrov, D.~Lewanski and S.~Shadrin,
  ``Ramifications of Hurwitz theory, KP integrability and quantum curves,''
  {\em JHEP} {\bf 1605}, 124 (2016).
  
\bibitem{AMMN1} 
  A.~Alexandrov, A.~Mironov, A.~Morozov and S.~Natanzon,
  ``Integrability of Hurwitz Partition Functions. I. Summary,''
 {\em  J.\ Phys.\ A} {\bf 45}, 045209 (2012).  

\bibitem{AMMN2}  A. Alexandrov, A. Mironov, A. Morozov  and S. Natanzon,   ``On KP-integrable Hurwitz functions,''
 {\em  JHEP} {\bf 1411}, 080 (2014).

 \bibitem{AC1} J. Ambj{\o}rn and L. Chekhov, ``The matrix model for dessins d'enfants'', 
{\em Ann. Inst. Henri Poincar\'e, Comb. Phys. Interact.} {\bf  1},  337-361 (2014).

\bibitem{AC2} J. Ambj{\o}rn and L. Chekhov, ``A matrix model for hypergeometric Hurwitz numbers'', 
{\em Theor. Math. Phys. } {\bf   181},  1486-1498 (2014).
  
 \bibitem{BEM}  R. Belliard, B. Eynard, O. Marchal, ``Loop equations from differential systems'', 
  ``Integrable differential systems of topological type and reconstruction by the topological recursion'',
   {\em Ann. Henri Poincar\'e} {\bf 18}  no. 10, 3193-3248 (2017). 

\bibitem{BC} E.A. Bender and E. Canfield, ``The number of rooted maps on an orientable surface'', {\it J. Combin. Theory}
 {\bf  Ser. B 53},  293-299 (1991).

\bibitem{BBE}  M. Bergere, G. Borot and B. Eynard, ``Rational differential systems, loop equations, and application to the q-th reduction of KP'',   {\it Ann. Henri Poincar\'e} {\bf 16},  2713-2782 (2015).
  
\bibitem{BEO}    G. Borot, B. Eynard, N. Orantin, ``Abstract loop equations, topological recursion, and applications'', 
{\it Commun. Num. Theor. Phys.} {\bf  09}  51-187 (2015).

\bibitem{BEMS} G. Borot, B. Eynard, M. Mulase and B. Safnuk, ``A matrix model for Hurwitz numbers and topological recursion'', {\em J. Geom. Phys.} {\bf 61}, 522--540 (2011). 

  \bibitem{BouchE}  V. Bouchard and B. Eynard, ``Think globally, Compute Locally'',  {\it JHEP} 
  {\bf  143},  1302 (2013).
  
 \bibitem{BMS} M. Bousquet-M\'elou and G. Schaeffer, ``Enumeration of planar constellations''.
{\it Adv. in Appl. Math.} {\bf 24:4}, 337-368 (2000).
  
  \bibitem{Ch} G. Chapuy, ``Asymptotic enumeration of constellations and related families of maps on orientable surfaces'', {\em Combin. Probab. Comput.} {\bf  18}, 477-566  (2009).
  
\bibitem{DDM}  N. Do, A. Dyer, and D.V. Mathews, ``Topological recursion and a quantum curve for monotone Hurwitz numbers", 
{\em J. Geom. Phys.} {\bf 120}, 19-36 (2017)

\bibitem{DLN} N. Do, O. Leigh, and P. Norbury, ``Orbifold Hurwitz numbers and Eynard-Orantin invariants", 
{\em Math. Res. Lett.} {\bf 23}, 1281-1327  (2016).

\bibitem{EO1} B.  Eynard and N. Orantin,  ``Invariants of algebraic curves and topological expansion'',
{\em Commun.  Number  Theor.  Phys.} {\bf 1}, 347-452 (2007).

\bibitem{EO2} B. Eynard  and N. Orantin, ``Topological recursion in enumerative geometry and random matrices'', 
{\em J. Phys. A} {\bf 42},  293001 (2009).

\bibitem{EO3} B. Eynard ``Invariants of spectral curves and intersection theory of moduli spaces of complex curves'',  {\em  Commun. Num. Theor. Phys.} {\bf 8}, 541-588  (2014). 

\bibitem{Fa} W. Fang, ``Enumerative and bijective aspects of combinatorial maps: generalization, unification and application'', PhD thesis, Universit\'e Paris Diderot (2016).
  
\bibitem{Frob}  G. Frobenius, ``\"Uber die Charaktere der symmetrischen Gruppe'', {\em Sitzber.  Akad. Wiss., Berlin},
516-534  (1900). Gesammelte Abhandlung III, 148-166;  {\it ibid.} ``\"Uber die Charakterische Einheiten der symmetrischen Gruppe'', {\em Sitzber.  Akad. Wiss., Berlin}, 328-358  (1903). Gesammelte Abhandlung III, 244-274.
 
 \bibitem{GGN1} I. P. Goulden, M. Guay-Paquet and J. Novak, ``Monotone Hurwitz numbers and the HCIZ Integral'', 
 {\em Ann. Math. Blaise Pascal} {\bf 21}, 71-99 (2014).
 
 \bibitem{GH1} M. Guay-Paquet and J. Harnad, ``2D Toda $\tau$-functions as combinatorial generating functions'', 
{\em Lett. Math. Phys.} {\bf 105}, 827-852 (2015).

\bibitem{GH2} M. Guay-Paquet and J. Harnad, ``Generating functions for weighted Hurwitz numbers'', 
{\em J. Math. Phys.} {\bf  58}, 083503 (2017).

\bibitem{H1} J. Harnad, ``Weighted Hurwitz numbers and hypergeometric $\tau$-functions: an overview'', {\it AMS Proceedings of Symposia in Pure Mathematics} {\bf  93},  289-333 (2016).

\bibitem{H2} J. Harnad, ``Quantum Hurwitz numbers and Macdonald polynomials'',
{\it J. Math. Phys.} {\bf  57}, 113505 (2016.) 

\bibitem{HO} J. Harnad and A. Yu. Orlov, ``Hypergeometric $\tau$-functions, Hurwitz numbers and enumeration of paths'',  {\em Commun. Math. Phys. } {\bf 338}, 267-284 (2015).

\bibitem{Hur} A. Hurwitz, ``\"Uber Riemann'sche Fl\"asche mit gegebnise Verzweigungspunkten'', 
{\em Math. Ann.} {\bf 39}, 1-61 (1891); Matematische Werke I, 321-384; {\it ibid.}  ``\"Uber die Anzahl der Riemann'sche Fl\"asche mit gegebnise Verzweigungspunkten'',  {\em Math. Ann.} {\bf 55}, 53-66 (1902); Matematische Werke I, 42-505.

\bibitem{Irving} J. Irving, ``On the enumeration of $m$-constellations'', (in preparation; private communication).

 \bibitem{KMMM}  S.~Kharchev, A.~Marshakov, A.~Mironov and A.~Morozov,
  ``Generalized Kazakov-Migdal-Kontsevich model: Group theory aspects,''
 {\em  Int.\ J.\ Mod.\ Phys.\ A}  {\bf 10}, 2015 (1995).
 
  \bibitem{KZ}  M. Kazarian and P. Zograf, ``Virasoro constraints and topological recursion for Grothendieck's
  dessin counting'',  {\it Lett. Math. Phys.} {\bf 105}, 1057-1084 (2015).
    
 \bibitem{LZ} S. K.  Lando and A.K.  Zvonkin {\em Graphs on Surfaces and their Applications}, Encyclopaedia of Mathematical Sciences, Volume {\bf 141}, with  appendix by D. Zagier, Springer, N.Y. (2004).

 \bibitem{Mac} I.~G. ~Macdonald, {\em Symmetric Functions and Hall Polynomials},
Clarendon Press, Oxford, (1995).

\bibitem{MSh}
  A.~Morozov and S.~Shakirov,
  ``Generation of Matrix Models by W-operators,''
  {\em JHEP} {\bf 0904} 064 (2009).

\bibitem{MSS} M.~Mulase, S.~Shadrin and L.~Spitz,
 `` The spectral curve and the Schr\"odinger equation of double Hurwitz numbers and higher spin structures,'' {\em Commun.\ Num.\ Theor Phys.}  {\bf 07}, 125 (2013).
 
 \bibitem{MS} 
  M.~Mulase and P.~Sulkowski,
  ``Spectral curves and the Schr\"odinger equations for the Eynard-Orantin recursion,''
   {\em Adv.\ Theor.\ Math.\ Phys.}  {\bf 19}, 955 (2015).
   
 \bibitem{NaOr} Sergey M. Natanzon and Aleksandr Yu. Orlov,
``BKP and projective Hurwitz numbers''
{\em Lett. Math. Phys.} {\bf 107},1065-1109 (2017).

\bibitem{Ok} A. Okounkov, ``Toda equations for Hurwitz numbers'', {\em Math.~Res.~Lett.} {\bf 7}, 447--453 (2000).

\bibitem{Or} A. Yu. Orlov, ``{Hypergeometric} functions as infinite-soliton $\tau$-functions'',  {\em Theor.  Math. Phys.} {\bf 146}, 183-206 (2006).

\bibitem{OrSc1} 
  A.~Y.~Orlov and D.~M.~Scherbin,
  ``Fermionic representation for basic hypergeometric functions related to Schur polynomials,''
  \arxiv{nlin/0001001}.
   
\bibitem{OrSc2} A. Yu. Orlov and D. M. Scherbin, ``Hypergeometric solutions of soliton equations'', {\em Theor.  Math. Phys.} {\bf 128}, 906-926 (2001).

\bibitem{Pa} R. Pandharipande, ``The Toda Equations and the Gromov-Witten Theory of
the Riemann Sphere'',  {\em Lett. Math. Phys.} {\bf  53}, 59-74 (2000).

\bibitem{Sa} M. Sato, ``Soliton equations as dynamical systems on infinite dimensional Grassmann 
manifolds'', RIMS, Kyoto Univ. {\it Kokyuroku} {\bf 439},  30-46 (1981).
 In ``Nonlinear Partial Differential Equations in Applied Science'';
 Proceedings of The U.S.-Japan Seminar, Tokyo, 1982, (ed. Hiroshi Fujita, Peter D. Lax, Gilbert Strang). 

  \bibitem{SS} M. Sato and Y. Sato. ``Soliton equations as dynamical systems on infinite dimensional 
Grassmann manifold'',  in: {\it Nonlinear PDE in Applied Science, Proc. U. S.-Japan Seminar}, Tokyo 1982, 
Kinokuniya, Tokyo, pp. 259-271 (1983).

\bibitem{Schaeffer} G. Schaeffer, ``Handbook of Enumerative Combinatorics'', Chapter 5 ,  
CRC Press {\it Discrete Mathematics and its Applications}(ed. Miklos Bon\`a)  (2015). 

 \bibitem{SW} G. Segal  and G. Wilson, ``Loop groups and equations of KdV type '', {\em Pub.
 Math. IH\'ES} {\bf 61}, 5-65 (1985).

\bibitem{Takeb} T. Takebe, ``Representation theoretical meaning of the initial value problem for the Toda lattice hierarchy I'', {\em Lett. Math. Phys.} {\bf 21}, 77--84,(1991).

\bibitem{Tutte} W. T. Tutte, ``A census of planar maps'', {\em Canad. J. Math.} {\bf 15} 249--271 (1963).  

\bibitem{UTa} K. Ueno and K. Takasaki, ``Toda Lattice Hierarchy'', in:{ \em Group Representation and Systems of Differential Equations}, 
{\em Adv. Stud. in Pure Math.} {\bf 4},  1-95 (1984). Mathematical Society of Japan, Tokyo, Japan (ed. K. Okamoto).

\bibitem{Z} P. Zograf, ``Enumeration of Gr\"othendieck's dessins and KP hierarchy'', {\it Int. Math. Res. Notices} 
{\bf 24},   13533-13544 (2015).  

\end{thebibliography}
\end{document}